\documentclass[letterpaper, 10 pt, conference]{ieeeconf}

\IEEEoverridecommandlockouts        % This command is only
\usepackage{cite}
\usepackage{amsmath,amssymb,amsfonts}
\usepackage{hyperref}
\usepackage{epstopdf,color}
\usepackage{textcomp}
\usepackage{graphicx,color}
\graphicspath{{img/}}
\usepackage{mathrsfs}
\usepackage[vlined,ruled]{algorithm2e}
\usepackage{subfigure}
\usepackage{url}
\usepackage{color}
\usepackage{dsfont}
\usepackage{bbm}
\usepackage{booktabs}
\usepackage{array}
\usepackage[table]{xcolor}
\usepackage{yfonts}
\usepackage{cleveref} % for the section symbol in the references
\usepackage{tikz}
\usepackage{pgfplots}
\usepackage{multirow}
\usepackage{textcomp}
\usepackage{tabularx}
\usepackage{placeins}
\usepackage{float}
\usepackage{nccmath}
\usepackage{accents}
\usepackage{multicol}
\usepackage{etoolbox, refcount}
\usepackage{chngcntr}
\usepackage{apptools}
\usepackage{relsize}
\usepackage[font=footnotesize]{caption}
\newtheorem{theorem}{Theorem}[section]
\newtheorem{lemma}[theorem]{Lemma}

\newtheorem{proposition}[theorem]{Proposition}

\newtheorem{remark}[theorem]{Remark}
\newtheorem{assumption}{Assumption}

\newcommand{\setdef}[2]{\{#1 \; : \; #2\}}

\newcommand{\Kc}{\mathcal{K}}

\newcommand{\Lc}{\mathcal{L}}

\newcommand{\real}{\mathbb{R}}

\newcommand{\Cc}{\mathcal{C}}
\newcommand{\Dc}{\mathcal{D}}
\newcommand{\Fc}{\mathcal{F}}

\newcommand{\Pc}{\mathcal{P}}

\newcommand{\Gc}{\mathcal{G}}
\newcommand{\Nc}{\mathcal{N}}

\newcommand{\Ec}{\mathcal{E}}

\DeclareSymbolFont{bbold}{U}{bbold}{m}{n}
\DeclareSymbolFontAlphabet{\mathbbold}{bbold}

\newcommand{\norm}[1]{\lVert#1\rVert}

\newcommand\oprocendsymbol{\hbox{$\bullet$}}
\newcommand\oprocend{\relax\ifmmode\else\unskip\hfill\fi\oprocendsymbol}

\renewcommand{\theenumi}{(\roman{enumi})}

\newcommand*{\QEDA}{\hfill\ensuremath{\blacksquare}}%

\newcommand\xqed[1]{%
  \leavevmode\unskip\penalty9999 \hbox{}\nobreak\hfill
  \quad\hbox{#1}}
\newcommand\demo{\xqed{$\bullet$}}

\newcounter{countitems}
\newcounter{nextitemizecount}
\newcommand{\setupcountitems}{%
  \stepcounter{nextitemizecount}%
  \setcounter{countitems}{0}%
  \preto\item{\stepcounter{countitems}}%
}
\makeatletter
\newcommand{\computecountitems}{%
  \edef\@currentlabel{\number\c@countitems}%
  \label{countitems@\number\numexpr\value{nextitemizecount}-1\relax}%
}
\newcommand{\nextitemizecount}{%
  \getrefnumber{countitems@\number\c@nextitemizecount}%
}
\newcommand{\previtemizecount}{%
  \getrefnumber{countitems@\number\numexpr\value{nextitemizecount}-1\relax}%
}
\makeatother    
{\end{itemize}%
\unskip\computecountitems\ifnumcomp{\previtemizecount}{>}{4}{\end{multicols}}{}}

\newcommand{\longthmtitle}[1]{\mbox{}\emph{(#1):}}

\renewcommand{\baselinestretch}{.98}
\newcommand{\comment}[1]{} % multiline comments
\newcolumntype{P}[1]{>{\centering\arraybackslash}p{#1}}
\allowdisplaybreaks

\begin{document}
\title{\bf Distributed Safe Navigation of Multi-Agent Systems using Control Barrier Function-Based Controllers }

\author{Pol Mestres \qquad Carlos Nieto-Granda \qquad Jorge Cort{\'e}s \thanks{P. Mestres and J. Cort\'es are with the Contextual Robotics Institute and the Department of Mechanical and Aerospace
    Engineering, UC San Diego, California,
    \{pomestre,cortes\}@ucsd.edu. Carlos Nieto-Granda is with the U.S. Army Combat Capabilities Development Command Army Research Laboratory (ARL), Adelphi, Maryland, carlos.p.nieto2.civ@army.mil.}%
}
\maketitle

%
% \marginJC{Need better title, b/c we do not use/define "distributed control barrier functions". Ideas we want in there are "distributed", "safe", "multi-agent system", "correct" (I mean a way to note we solve the opt problem exactly), something that alludes to the "coupled CBF-based optimization formulation/problem".}
%

\begin{abstract}
This paper proposes a distributed controller synthesis framework for safe navigation of multi-agent systems. We leverage control barrier functions to formulate collision avoidance with obstacles and teammates
as constraints on the control input for a state-dependent network optimization problem that encodes team formation and the navigation task. Our algorithmic solution is valid under general assumptions for nonlinear dynamics and state-dependent network optimization problems with convex constraints and strongly convex objectives. The resulting controller is distributed, satisfies the safety constraints at all times, and asymptotically converges to the solution of the state-dependent network optimization problem. We illustrate its performance in a team of differential-drive robots in a variety of complex environments, both in simulation and in hardware.
\end{abstract}

\section{Introduction}\label{sec:introduction}

Safety-critical control has received a lot of attention in the robotics and controls communities motivated by applications like autonomous driving, navigation of robotic swarms, and optimal power flow of energy grids,
where safety constraints are ubiquitous. Control barrier functions (CBFs) are a computationally efficient tool for the synthesis of safe controllers. In scenarios involving multiple autonomous agents, safety constraints often couple the decisions of the different team members. The challenge we face then is how to extend the CBF framework to multi-agent settings in a distributed manner (i.e., with each agent designing its local controller using only information from neighboring agents) while still retaining its efficiency, safety, and optimality guarantees. This question serves as the main motivation for this paper.

\subsubsection*{Literature Review}
This work draws on notions and techniques from the body of work
pertaining to CBFs~\cite{ADA-SC-ME-GN-KS-PT:19, PW-FA:07}, which are a powerful tool
to render safe a given set. Formally, safe controllers can be synthesized~\cite{ADA-XX-JWG-PT:17} by incorporating CBF-based conditions as constraints in an optimization problem whose objective encodes desired specifications on the input. 
% which is a Quadratic Program (QP) if the dynamics are control-affine, whose solution defines a safe controller. 
% If the dynamics are control-affine, such optimization problem is a Quadratic Program (QP) and can be solved efficiently~\cite{BS-GB-PG-AB-SB:20}.
Such optimization problems have state-dependent constraints and need to be studied in feedback loop with a plant. 
% This type of optimization-based controllers have been well-studied in the literature~\cite{AH-SB-GH-FD:21,MC-EDA-AB:20,BJM-MJP-ADA:15,PM-AA-JC:23-scl}.
Often, these problems cannot be solved instantaneously and instead approximate controllers must be implemented. Here, we draw on ideas from the ``dynamical systems approach to algorithms" that views optimization algorithms as continuous-time dynamical systems, cf.~\cite{RWB:91,UH-JBM:94}, which is a useful perspective when implementing optimization-based controllers in feedback systems~\cite{AH-SB-FD:21,MC-EDA-AB:20,MC-JSP-AB:19,AA-JC:24-tac,AC-YC-BC-JC-EDA:24-tsg}.
%
% \marginJC{Only citing the ETH people here seems unfair (and might infuriate lots of potential reviewers): we don't want 10 cites, but we need to cast a wider net}
%
In the context of multi-agent systems, many applications require a distributed implementation of the individual agents' controllers. For optimization-based controllers resulting from CBFs, this requires a distributed solution of the resulting optimization problem. The works~\cite{UB-LW-ADA-ME:15, LW-ADA-ME:17} tackle this problem 
%also \cite{MJ-MS:21}
for quadratic programs (QPs) by splitting a centralized QP into local QPs that can be solved efficiently while preserving safety guarantees. However, the solution of these local QPs might lack optimality guarantees with respect to the original QP. The recent works~\cite{XT-DVD:22, VNFA-XT-DVD:23, XT-CL-KHJ-DVD:24, CL-XT-XW-DVD-KHJ:24, PM-JC:23-cdc}
introduce different algorithms to solve constrained optimization problems in a distributed fashion while satisfying the constraints throughout the execution of the algorithm. However,~\cite{XT-DVD:22,XT-CL-KHJ-DVD:24,VNFA-XT-DVD:23,CL-XT-XW-DVD-KHJ:24} are restricted to a class of parametric QPs with conditions on the coefficients of the affine constraints. This approach has recently been successfully implemented for the problem of global connectivity maintenance~\cite{NDC-PS-PRG:24}.
%
% \marginJC{This wording is a bit generic: what "certain class"? We don't want to get technical, but we need more specifics.}
%
On the other hand, our previous work~\cite{PM-JC:23-cdc} does not consider the implementation of the optimization problem in feedback loop with a plant.

\subsubsection*{Statement of Contributions}

We design a distributed controller for safe navigation of multi-agent systems. 
We propose a synthesis framework which leverages CBFs to formulate obstacle avoidance and inter-agent collision avoidance constraints as affine inequalities in the control input. These constraints are included in a state-dependent network optimization problem that finds the control inputs allowing the agents to reach different waypoints of interest while maintaining a given formation and satisfying the safety constraints. Our first contribution leverages the particular structure of the optimization problem to decouple its constraints by using a set of auxiliary variables while still keeping the same feasible set. This enables us to derive a distributed update law for these auxiliary variables that allows each agent to obtain its local control input by solving a local optimization problem without the need to coordinate with any other agent. Our second contribution establishes that the proposed controller design is distributed, safe, and asymptotically converges to the solution of the state-dependent network optimization problem. Our last contribution is the implementation of the proposed controller in a variety of different environments, robots and formations, both in simulation and in real hardware.

%
% \marginJC{Pol to win space by streamlining presentation to remove widow lines}
%

\section{Preliminaries}\label{sec:prelims}

We introduce here basic notation and preliminaries on CBFs, constrained optimization, and projected saddle-point dynamics. This section can be safely skipped by a reader already familiar with these notions.

\subsubsection*{Notation}
We denote by $\real$ and $\real_{\geq 0}$ the set of real and nonnegative real numbers, respectively. For a positive integer $n$, we let $[n]=\{1,\hdots,n\}$. Given a set $A$, $|A|$ denotes its cardinality. Given $x\in\real^n$, $\norm{x}$ denotes its Euclidean norm. For $a\in\real$ and $b\in\real_{\geq0}$, we let
\begin{align*}
  [a]_b^+ = \begin{cases}
    a, \quad &\text{if} \ b>0, \\
    \max\{0,a\}, \quad &\text{if} \ b=0.
\end{cases}
\end{align*}
Given $a\in\real^n$ and $b\in\real_{\geq0}^n$, $[a]_b^{+}$ denotes the vector whose $i$-th component is $[a_i]_{b_i}^{+}$, for $i\in[n]$.
Given a set
$\Pc\subseteq\real^{n}$ and variables
$\xi=\{ x_{i_j} \}_{j=1}^k$, we denote by
$\Pi_{\xi}\Pc=\setdef{(x_{i_1},x_{i_2},\dots,x_{i_k})\in\real^k}{x\in\Pc}$
the projection of $\Pc$ onto the $\xi$ variables. An undirected graph is a pair $\Gc=(V,\Ec)$, where $V=V(\Gc)=[N]$ is the vertex set and $\Ec = \Ec(\Gc)\subseteq V\times V$ is the edge set, with $(i,j)\in\Ec$ if and only if $(j,i)\in\Ec$.
%Given a graph $\Gc$, $V(\Gc)$ denotes the set of vertices of $\Gc$. 
The set $\Nc_i=\setdef{j\in V}{(i,j)\in\Ec}$ denotes the neighbors of node $i$.
For a function $f:\real^n\times\real^m\to\real$, we denote by $\nabla_x f$ and $\nabla_y f$ the column vectors of partial derivatives of $f$ with respect to the first and second arguments, respectively.
Let $F:\real^{n}\to\real^{n}$ be locally Lipschitz. A set $\Cc$ is forward invariant for the dynamical system $\dot{x}=F(x)$ if all trajectories that start in $\Cc$ remain in $\Cc$ for all positive times.
A function $\alpha :\real\to\real$ is of extended class $\Kc_{\infty}$ if $\alpha(0) = 0$, $\alpha$ is strictly increasing, and $\lim\limits_{t\to\infty}\alpha(t)=\infty$.

%For $t\ge 0$, we denote its flow map by
% $\Phi_{t}:\real^{n}\to\real^{n}$, defined by $\Phi_{t}(x)=x(t)$,
% where $x(t)$ is the unique solution of the dynamical system with
% $x(0)=x$. A set $\Kc\subseteq\real^{n}$ is forward
% invariant if $x\in\Kc$ implies that $\Phi_{t}(x)\in\Kc$ for all
% $t\geq0$.

\subsubsection*{Control Barrier Functions}
Consider the control-affine system
\begin{align}\label{eq:control-affine-sys}
    \dot{x}=F(x)+G(x)u
\end{align}
where $F:\real^n\rightarrow\real^n$ and $G:\real^n\to\real^{n\times m}$ are locally Lipschitz functions, with $x\in\real^n$ the state and $u\in\real^m$ the input. Let $\Cc\subset\real^n$ and $h:\real^n\to\real$ be a continuously differentiable function such that
\begin{align}\label{eq:safe-set}
  \Cc \!= \! \setdef{x\in\real^n \! \! }{ \! h(x)\geq0}, \
  \partial\Cc \!=\! \setdef{x\in\real^n \! \!}{ \! h(x)=0}.
\end{align}
The function $h$ is a control barrier function (CBF)~\cite[Definition 2]{ADA-SC-ME-GN-KS-PT:19} of $\Cc$ if there exists an extended class $\Kc_{\infty}$ function $\alpha$ and $\Dc\subset\real^n$ with $\Cc\subset\Dc$ such that for all $x\in\Dc$, there exists $u\in\real^m$ with
\begin{align}\label{eq:cbf-ineq}
L_F h(x) + L_G h(x)u +\alpha(h(x)) \geq 0.
\end{align}
A Lipschitz controller $k:\real^n\to\real^m$ such that $u=k(x)$ satisfies~\eqref{eq:cbf-ineq} for all $x\in\Cc$ makes $\Cc$ forward invariant. Hence, CBFs provide a way to guarantee safety.

\subsubsection*{Projected Saddle-Point Dynamics}
Given continuously differentiable functions $f:\real^n\to\real$ and $g:\real^n\to\real^p$, whose derivatives are locally Lipschitz, consider the constrained nonlinear program
\begin{align}\label{eq:constrained-opt}
  & \min \limits_{x\in\real^n} f(x)
  \\
  \notag
  \qquad & \text{s.t.} \quad g(x) \leq0.
\end{align}
Let $\Lc:\real^n\times\real^p_{\geq0}\to\real$ be its associated Lagrangian, $\Lc(x,\lambda) = f(x) + \lambda^T g(x)$.
The \textit{projected saddle-point dynamics} for $\Lc$ are defined as follows:
\begin{subequations}
  \begin{align}
    \dot{x} &= -\nabla_x \Lc(x,\lambda,\mu),
    \\
    \dot{\lambda} &= [\nabla_{\lambda}
                    \Lc(x,\lambda,\mu)]_{\lambda}^{+}.
  \end{align}
  \label{eq:projected-saddle-point-dyn}
\end{subequations}
If $f$ is strongly convex and $g$ is convex, then  $\Lc$ has a  unique saddle point, which corresponds to the KKT point of~\eqref{eq:constrained-opt}. Moreover,~\cite[Theorem 5.1]{AC-EM-SHL-JC:18-tac} ensures that it is globally asymptotically stable under the dynamics~\eqref{eq:projected-saddle-point-dyn}.

\subsubsection*{Constraint Mismatch Variables}
%Here we introduce the notion of constraint mismatch variables, which can be used to give an equivalent formulation for a network optimization problem with separable objective function and constraints.
Given $N\in\mathbb{Z}_{>0}$, 
consider a network composed by agents $\{1,\hdots,N\}$ whose communication topology is described by a connected undirected graph $\Gc$. An edge $(i,j)$ represents the fact that agent $i$ can receive information from agent $j$ and vice versa.
For each $i\in[N]$ and $k\in[p]$, where $p\in\mathbb{Z}_{>0}$,
let $f_i:\real^n\to\real$ be a strongly convex function and $g_i^k:\real^n \rightarrow \real$ a convex function.
We consider the following optimization problem with separable objective function and constraints
\begin{align}\label{eq:separable-distr-opt-pb}
  & \min \limits_{ u \in\real^{nN} } \sum_{i=1}^N f_i( u_i)
  \\
  \notag
  \qquad & \text{s.t.} \quad \sum_{i\in V(\Gc_k)} g_i^{k}( u_i ) \leq0, \quad
           k\in[p],
\end{align}
where $\Gc_k$ is a connected subgraph of $\Gc$ for each $k\in[p]$. 
% The case with equality constraints in problem~\eqref{eq:separable-distr-opt-pb} can be adapted following a similar argument to the one presented here.
For each $k \in [p]$, the constraint in~\eqref{eq:separable-distr-opt-pb} couples the local variable~$u_i$ of agent $i$ with the variables of all the other agents in $\Gc_k$. 
The constraints in~\eqref{eq:separable-distr-opt-pb} can be decoupled by introducing \textit{constraint-mismatch variables}~\cite{AC-JC:16-allerton,PM-JC:23-cdc}, which help agents keep track of local constraints while collectively satisfying the original constraints.
Specifically, we add one constraint-mismatch variable per agent and constraint.
We let $z_i^k\in\real$ be the constraint-mismatch variable for agent $i$ and constraint $k$. Let $P_i:=
\setdef{k\in[p]}{i\in V(\Gc_k)}$ be the indices of the constraints involving agent $i$. For convenience, we let $u=[u_1,\hdots,u_{N}]$, $z_i=\{ z_i^k \}_{k\in P_i}$, $z=[z_1,\hdots,z_N]$ and $q:=\sum_{i=1}^N |P_i|$.
Next, consider the problem
\begin{align}\label{eq:distr-opt-pb-y}
  & \min \limits_{u \in\real^{nN}, \ z \in\real^q } \sum_{i=1}^N f_i(u_i)
  \\
  \notag
  & \text{s.t.} \quad g_i^k(u_i)+\sum_{j\in
    \Nc_i \cap V(\Gc_k)}(z_i^k-z_j^k)\leq0, \ i\in V(\Gc_k), \ k\in [p].
\end{align}
In~\eqref{eq:distr-opt-pb-y}, the constraints are now locally expressible, meaning that agent $i\in[N]$ can evaluate the ones in which its  variable $u_i$ is present by using variables obtained from its neighbors.
The next result establishes the equivalence of~\eqref{eq:separable-distr-opt-pb} and ~\eqref{eq:distr-opt-pb-y}.

\begin{proposition}\longthmtitle{Equivalence between the two formulations}\label{prop:equivalence}
    Let $\Fc^*$ be the solution set of~\eqref{eq:distr-opt-pb-y}. Then, $u^*=\Pi_u{\color{blue}}(\Fc^*)$ is the optimizer of~\eqref{eq:separable-distr-opt-pb}.
\end{proposition}

The proof is similar to~\cite[Proposition 4.1]{PM-JC:23-cdc}, with the necessary modifications to account for the fact that the constraints of~\eqref{eq:separable-distr-opt-pb} might only involve a subset of the agents.

\section{Problem Statement}\label{sec:problem-statement}
We are interested in designing distributed controllers that allow a
team of differential-drive robots to safely navigate an environment
while maintaining a desired formation and visiting a sequence
of waypoints of interest.  The robots have identities
$\{ 1,\hdots,N \}$ and follow unicycle dynamics
\begin{align}\label{eq:unicycle-dynamics}
    \dot{x}_i = v_i \cos{\theta_i}, \quad \dot{y}_i = v_i \sin{\theta_i}, \quad
    \dot{\theta_i} = \omega_i,
\end{align}
where $s_i=[x_i,y_i]\in\real^2$ is the position of agent $i$, $\theta_i\in[0,2\pi)$ its heading and $v_i\in\real$ and $w_i\in\real$ are its linear and angular velocity control inputs, respectively.

We next leverage CBFs to encode the different safety specifications, giving rise to affine constraints in the control inputs. We note that, under the dynamics~\eqref{eq:unicycle-dynamics}, direct application of~\eqref{eq:cbf-ineq} on functions that only depend on position results in limited design flexibility because $\omega_i$ does not appear in the time derivative of $s_i$. Instead, we follow~\cite[Section IV]{PG-IB-ME:19} and define, for all $i\in[N]$,
\begin{align*}
    R(\theta_i) \! = \! \begin{bmatrix}
        \cos{\theta_i} & \! -\sin{\theta_i} \\
        \sin{\theta_i} & \! \cos{\theta_i}
    \end{bmatrix}, \ p_i \! = \! s_i \! + \! l R(\theta_i) e_1, \ L \! = \! \begin{bmatrix}
        1 & \! 0 \\
        0 & \! 1/l
    \end{bmatrix} 
\end{align*}
where $e_1=[1, 0]^T$ and $l>0$ is a design parameter. This defines $p_i$ as a point orthogonal to the wheel axis of the robot. 
%(cf.~\cite[Figure 1]{PG-IB-ME:19}).
It follows that $\dot{p}_i=R(\theta_i)L^{-1}u_i$, where $u_i=[v_i,w_i]^T$. Hence, by choosing $p_i$ as our position variable, both control inputs appear in the time derivative of the position, which allows for more flexibility in our control design.

\subsubsection*{Avoiding obstacles}
We consider an environment with $M$ obstacles $\{ \mathcal{O}_k \}_{k\in[M]}$, each expressable as the sublevel set of a differentiable function $h_k$, i.e., $\mathcal{O}_k =\setdef{(x,y)\in\real^2}{h_k(x,y) < 0}$. To guarantee that the robot does not collide with the obstacles, we want to ensure that
\begin{align}\label{eq:hk-p}
    h_k(p_i)\geq\eta_k > 0, \quad \forall k\in[M].
\end{align}
The parameter $\eta_k\in\real$ should be taken large enough to
guarantee that the whole physical robot (instead of just $p_i$) does
not collide with the obstacle $\mathcal{O}_k$.  For instance, if $r_i$
is the radius of robot $i$ and $\mathcal{O}_k$ is a circular obstacle
with center at $m_k\in\real^2$ and radius $R_k$ so that
$h_k(p_i)=\norm{p_i-m_k}^2-R_k^2$, then $\eta_k$ can be taken as
$(r_i+l)^2 + 2R_k(r_i+l)$.  The CBF condition associated
with~\eqref{eq:hk-p} with linear extended class $\Kc_{\infty}$ function with slope $\alpha_k>0$ reads
\begin{align}\label{eq:hk-p-cbf}
    \nabla h_k(p_i)^T R(\theta_i) L^{-1} u_i \geq -\alpha_k(h_k(p_i)-\eta_k) .
\end{align}

\subsubsection*{Avoiding inter-agent collisions}
We also want to enforce that agents do not collide with other team members. To achieve this, 
we assume it is enough for each robot $i$ to avoid colliding with a subset of the agents $N_i\subset [N] \backslash \{ i \}$.
This is motivated by the fact that our control design will make the team maintain a formation at all times. Hence, agent only needs to avoid colliding with agents closest to it in the formation.
% which means that the set of agents closest to any given agent will mostly remain constant in time.
For this reason, we assume that if $j\in N_i$, then $i\in N_j$, and agent $i$ can communicate with all agents in $N_i$ to obtain their state variables. 
% This also defines the communication graph, i.e., an edge exists between agents $i$ and $j$ if and only if $i\in N_j$ and $j\in N_i$. 
We assume that the resulting communication graph is connected.
 %
% \marginJC{So this defines the comm graph? How do we know it is connected (I ask b/c in the next section we assume it to be)? Why do we need it connected? The comm graph should also enable communication of every agent with the leader?}
%
For any $i\in[N]$ and $j\in N_i$, we want to ensure
\begin{align}\label{eq:dpi-pj}
    d(p_i,p_j):=\norm{p_i-p_j}^2 - d_{\text{min}}^2 \geq 0,
\end{align}
with $d_{\text{min}}\geq r_i+r_j+2l$. The CBF condition for~\eqref{eq:dpi-pj} with a linear extended class $\Kc_{\infty}$ function with slope $\alpha_c^{ij}>0$ reads
\begin{align}\label{eq:interagent-coll-avoidance-cbf}
    2(p_i\!-\!p_j)^T (R(\theta_i) L^{-1} u_i \!-\! R(\theta_j) L^{-1}
  u_j) \geq - \alpha_c^{ij} d(p_i,p_j).
\end{align}

\subsubsection*{Team formation}
Finally, we are interested in making the team reach a goal while
maintaining a certain formation. To do so, we define a \textit{leader}
for the team (without loss of generality, agent $1$). The rest of the
agents $\{2,\hdots,N\}$ are referred to as \textit{followers}.  The
leader is in charge of steering the team towards a given waypoint
$q_1\in\real^2$.  To achieve it, we define a nominal controller
$u_{\text{nom},1}:\real^2\to\real^2$ that steers it towards $q_1$. We use the stabilizing controller for the unicycle dynamics
in~\cite{MA-GC-AB-AB:95}. To define it, let $e_1:\real^2\to\real$ and $\beta_1:\real^2\to\real$ defined as
\begin{align*}
    e_1(p_1) = \norm{p_1-q_1}, \quad \beta_1(p_1) = \arctan \Big(\frac{p_{1,y}-q_{1,y}}{p_{1,x}-q_{1,x}} \Big).
\end{align*}
The control law is $u_{\text{nom},1}(p_1)=[v_1(p_1),w_1(p_1)]$, where
\begin{subequations}
\begin{align}
    v_1(p_1) &= k_r e(p_1) \cos( \beta_1(p_1) )-\theta_1 ), 
    \\
    \omega_1(p_1) &= k_a \beta_1(p_1) +
    \frac{k_r}{2} \sin( 2\beta_1(p_1) )
    \frac{\beta_1(p_1)\!+\!h\theta_1}{\beta_1(p_1)},
\end{align}
\label{eq:unicycle-nominal-stabilizing}
\end{subequations}
and $k_r > 0, k_a > 0$ and $h > 0$ are design parameters. We can also specify a sequence of waypoints for the leader: once the leader is within a given
tolerance of the current waypoint,
$u_{\text{nom}, 1}$ can be updated to steer it towards the
next waypoint.

As the leader moves towards the desired waypoint, the followers follow it while maintaining a certain formation of interest. We define the desired formation positions for the followers as follows.
First, recall that $N_1 :=\setdef{i\in[N]\backslash\{ 1 \} }{i \ \text{and} \ 1 \ \text{can communicate their respective state variables}}$ and, for $k\in\mathbb{Z}_{>0}$, $k>1$, define the $k$-neighborhood of the leader as $N_1^k:=\setdef{i\in[N]}{\exists \ j\in N_1^{k-1} \ \text{s.t.} \ i\in N_j}$ (here, $N_1^1=N_1$). Since the communication graph is connected, there exists $K\in\mathbb{Z}_{>0}$ such that, for all $i\in[N]\backslash\{ 1 \}$, there is $k\in[K]$ such that $i\in N_1^k$. For every $i\in[N]\backslash \{ 1 \}$, we let 
$k_i$ be the smallest positive integer $k$ such that $i\in N_1^k$.
For every $i\in[N]\backslash\{1\}$, we consider functions $q_i:\real^{2|N_1^{k_i-1}\cap N_i|}\to\real^2$ such that $q_i( \{ p_j \}_{j\in N_1^{k_i-1}\cap N_i} )$ defines the desired formation position for agent $i$. Agent $i$ aims to remain as close as possible to $q_i( \{ p_j \}_{j\in N_1^{k_i-1}\cap N_i} )$ while maintaining the safety constraints. To achieve this, we define a nominal controller $u_{\text{nom},i}:\real^2\to\real^2$ for each agent
$i\in\{2,\hdots,N\}$ analogous to~\eqref{eq:unicycle-nominal-stabilizing} that steers it towards $q_i( \{ p_j \}_{j\in N_1^{k_i-1} \cap N_i} )$. Note that $q_i$ can be computed in a distributed fashion because it only depends on the positions of agents in $N_1^{k_i-1}\cap N_i \subseteq N_i$.

Agents collectively try to design controllers $\{ u_i \}_{i=1}^N$
that satisfy the obstacle avoidance and inter-agent collision
avoidance constraints while remaining as close as possible to
their nominal controller. By employing weighting matrix-valued functions $\Gamma_i:\real^n\to\real^{m\times m}$ for each $i\in[N]$ that can be designed to penalize the linear and angular velocity inputs differently, and leveraging the CBF conditions~\eqref{eq:hk-p-cbf},~\eqref{eq:interagent-coll-avoidance-cbf}, we obtain the following network-wide optimization problem:
\begin{align}\label{eq:diff-drive-pb}
  &\min \limits_{ \{u_i \in \real^2 \}_{i=1}^N } \quad \sum_{i=1}^N \frac{1}{2}\norm{ \Gamma_i(p_i)( u_i-u_{\text{nom}, i}(p_i) )}^2
  \\
  \notag
  & \text{s.t.} \ (p_i-p_j)^T (R(\theta_i) L^{-1} u_i - R(\theta_j)L^{-1}u_j) \geq -\alpha_c^{ij} d(p_i,p_j), \\
  \notag
  & \quad \nabla h_k(p_i)^T R(\theta_i) L^{-1} u_i \geq -\alpha_k (h_k(p_i)-\eta_k), \\
  \notag
  & \qquad \qquad \qquad \qquad \qquad \qquad i \in [N], j\in N_i, \ k\in [M].
\end{align}
The presence of the controls $u_i$ and $u_j$ in the inter-agent collision avoidance
constraints~\eqref{eq:interagent-coll-avoidance-cbf} means that the satisfaction of such constraints requires
coordination between agents involved in the
constraint.
This, together with the state-dependency of the objective
function and constraints, poses challenges in the implementation of a
distributed algorithm that solves~\eqref{eq:diff-drive-pb}.  Our goal
is to design a distributed controller 
% (i.e.,
% such that each agent can execute it only with information from itself
% and agents from $N_i$) 
that satisfies the constraints at all times and with the same optimality properties as the solution obtained directly from~\eqref{eq:diff-drive-pb}.

\section{Distributed Controller Design}\label{sec:proposed-solution}
%\section{Proposed Solution: Interconnection of Projected Saddle-Point Dynamics and Plant}\label{sec:proposed-solution}

In this section we design a distributed algorithmic solution to the problem formulated in
Section~\ref{sec:problem-statement}. As it turns out, our solution is valid for a more general setup, as we explain next.
Assume that the agents' communication network is described by a connected
undirected graph $\Gc$, as in Section~\ref{sec:problem-statement}.
An edge $(i,j)$ represents the fact
that agent $i$ can receive information from agent $j$ and vice
versa. We describe the dynamics of each agent $i\in[N]$ by
\begin{align}\label{eq:general-dynamics}
    \dot{\xi}_i = F_i(\xi_i,u_i)
\end{align}
where $F_i:\real^n\times\real^m\to\real^n$ is a locally Lipschitz function for each $i\in[N]$, $\xi_i\in\real^n$ is the state variable of agent $i$ and $u_i\in\real^m$ is its local control input.
Additionally, let $f_i:\real^n\times\real^m\to\real$ and $g_i^k:\real^{n|\Nc_i|}\times\real^m\to\real$, with $i\in[N]$ and $k\in[p]$ be functions satisfying the following.

\begin{assumption}\longthmtitle{Regularity and convexity of the optimization problem}\label{as:regularity-convexity-optimization-pb}
    For all $i\in[N]$ and $k\in[p]$, $f_i$ and $g_i^k$ are continuously differentiable functions with Lipschitz derivatives. 
    We assume that for all $i\in[N]$ and $\xi_i\in\real^n$, the functions $f_i(\xi_i,\cdot)$ are strongly convex, and for all $i\in[N]$, $k\in[p]$ and $\bar{\xi}_i\in\real^{n|\Nc_i|}$, the functions $g_i^k(\bar{\xi}_i,\cdot)$ are convex. 
\end{assumption}
The agents need to coordinate to solve an optimization problem of the form
\begin{align}\label{eq:distr-opt-pb}
  & \min \limits_{{\{u_i \in \real^m \}_{i=1}^N}} \sum_{i=1}^N f_i(\xi_i, u_i),
  \\
  \notag
    & \qquad \text{s.t.} \sum_{i\in V(\Gc_k)} g_i^k(\xi_{\Nc_i}, u_i) \leq0, \quad k \in [p].
\end{align}
where $\xi_{\Nc_i}=\{\xi_j\}_{j\in\Nc_i}$ and $\Gc_k$ is a connected subgraph of $\Gc$ for each $k\in[p]$.
Note that~\eqref{eq:diff-drive-pb} is a particular case of~\eqref{eq:distr-opt-pb},  where 
$f_i(\xi_i,u_i)=\frac{1}{2}\norm{\Gamma_i(\xi_i)(u_i-u_{\text{nom},i}(\xi_i) )}^2$ and $\Gc_k$ corresponds to the graph with nodes $\{ i \} \cup \{ j \}$ and an edge between $\{ i \}$ and $\{ j \}$ for the inter-agent collision avoidance constraint between agents $i$ and $j$, and $\Gc_k$ corresponds to the singleton $\{ i \}$ for the obstacle avoidance constraints of agent~$i$. Moreover, for the inter-agent collision avoidance constraint between agents $i$ and $j$,
{\small
\begin{align*}
g_i^k(p_i,\theta_i,p_j,\theta_j,u_i) &=-2(p_i-p_j)^T R(\theta_i) L^{-1} u_i - \alpha_c^{ij} d(p_i,p_j) ,
\\
g_j^k(p_j,\theta_j,p_i,\theta_i,u_j) &=-2(p_j-p_i)^T R(\theta_j) L^{-1} u_j - \alpha_c^{ij} d(p_i,p_j).    
\end{align*}
}
For the collision avoidance constraint of agent $i$ with the $k$th obstacle, $g_i^k(p_i,\theta_i,u_i) = -\nabla h_k(p_i)^T R(\theta_i) L^{-1}u_i - \alpha_k (h_k(p_i)-\eta_k)$.
Problem~\eqref{eq:distr-opt-pb} can encode other types of safety constraints and more general dynamics.
We henceforth denote $\xi=[\xi_1,\hdots,\xi_N]\in\real^{nN}$ and $u=[u_1,\hdots,u_N]\in\real^{mN}$. 
We assume that~\eqref{eq:distr-opt-pb} is feasible.

\begin{assumption}\label{as:feasibility-no-constraint-mismatch-vars}
    Problem~\eqref{eq:distr-opt-pb} is feasible for all $\xi\in\real^{nN}$.
\end{assumption}
\smallskip

Assumption~\ref{as:feasibility-no-constraint-mismatch-vars} is necessary for the solution of~\eqref{eq:distr-opt-pb} to be well defined for all $\xi \in R^{nN}$. If the constraints of~\eqref{eq:distr-opt-pb} are defined by CBFs, such as those in~\eqref{eq:diff-drive-pb}, their joint feasibility can be characterized, cf.~\cite{PM-JC:23-csl,XX:18,XT-DVD:22-cdc}.

We let $u^*:\real^{nN}\to\real^{mN}$ be the function mapping each $\xi\in\real^{nN}$ to the solution of~\eqref{eq:distr-opt-pb}. To tackle the design of our distributed algorithm, we first deal with the coupling induced by the inequality constraints by introducing \textit{constraint mismatch variables} $z_i^k$ for each agent and constraint in which it is involved. We use the same notation as in Section~\ref{sec:prelims}
and reformulate~\eqref{eq:distr-opt-pb} as
\begin{align}\label{eq:cbf-distr-opt-pb-y}
  & \min \limits_{u\in\real^{mN}, z\in\real^q} \sum_{i=1}^N f_i(\xi_i, u_i),
  \\
  \notag
    & \text{s.t.} \quad g_i^k(\xi_{\Nc_i}, u_i)+ \! \! \! \! \! \! \! \sum_{j\in \Nc_i\cap V(\Gc_k)} \! \! \! \! (z_i^k \! - \! z_j^k) \leq0, \ i \! \in \! V(\Gc_k), \ k \! \in \! [p].
\end{align}
In order to facilitate the analysis of the convergence properties of the algorithms that will follow, we regularize~\eqref{eq:cbf-distr-opt-pb-y} by adding the term $\epsilon \sum_{k=1}^p \sum_{j\in V(\Gc_k)} (z_j^k)^2$, with $\epsilon>0$, in the objective function of~\eqref{eq:cbf-distr-opt-pb-y}. Details regarding this regularization are covered in the Appendix.
With the added regularization term,  the problem~\eqref{eq:cbf-distr-opt-pb-y} has a strongly convex objective function and convex constraints, and therefore has a unique optimizer for every $\xi \in\real^{nN}$.
Let $u^{*,\epsilon}:\real^{nN}\to\real^{mN}$ and $z^{*,\epsilon}:\real^{nN}\to\real^q$ be the functions mapping each $\xi\in\real^{nN}$ to the corresponding unique optimizers in $u$ and $z$, respectively, of the regularized problem. We also let $\lambda^{*,\epsilon}:\real^{nN}\to\real^q$ be the function mapping each $\xi\in\real^{nN}$ to the optimal Lagrange multiplier of the regularized problem.
Problem~\eqref{eq:cbf-distr-opt-pb-y} (or its regularized version) cannot be decoupled into $N$ local optimization problems because the
optimal values of $z$ in~\eqref{eq:cbf-distr-opt-pb-y} require coordination among agents.
However, for every fixed $z$,~\eqref{eq:cbf-distr-opt-pb-y} can be solved locally by agent $i$ by just optimizing over $u_i$ as follows:
\begin{align}\label{eq:cbf-distr-opt-pb-y-local-only-u}
  & \min \limits_{u_i\in\real^m} f_i(\xi_i, u_i),
  \\
  \notag
    & \text{s.t.} \quad g_i^k(\xi_{\Nc_i}, u_i)+\sum_{j\in \Nc_i \cap V(\Gc_k)}(z_i^k-z_j^k) \leq0, \quad \ k\in[p].
\end{align}
We let $\bar{u}_i:\real^{n|\Nc_i|}\times\real^{|\Nc_i|}\to\real^m$ be the function that maps every $(\xi_{\Nc_i},z_{\Nc_i})\in\real^{n |\Nc_i|}\times\real^{\sum_{j\in\Nc_i}|P_j|}$ to 
the optimizer of~\eqref{eq:cbf-distr-opt-pb-y-local-only-u}, and $\bar{u}=[\bar{u}_1,\hdots,\bar{u}_N]$.
By construction, the controller $\bar{u}_i$ is \textbf{distributed}, as it only depends on variables which can be obtained through communication with agents in $\Nc_i$.

However, $\bar{u}_i$ depends on the chosen value of $z$. Since~\eqref{eq:cbf-distr-opt-pb-y} contains a minimization over both $u$ and $z$, $\bar{u}_i$ coincides with the optimizer over $u$ of~\eqref{eq:distr-opt-pb}, one must also optimize over the \textit{constraint mismatch} variables~$z$. We do this by updating them with the projected saddle-point dynamics of the regularization of~\eqref{eq:cbf-distr-opt-pb-y}, cf. Figure~\ref{fig:scheme}.
\begin{figure}
    \centering
    \includegraphics[width = 0.4\textwidth]{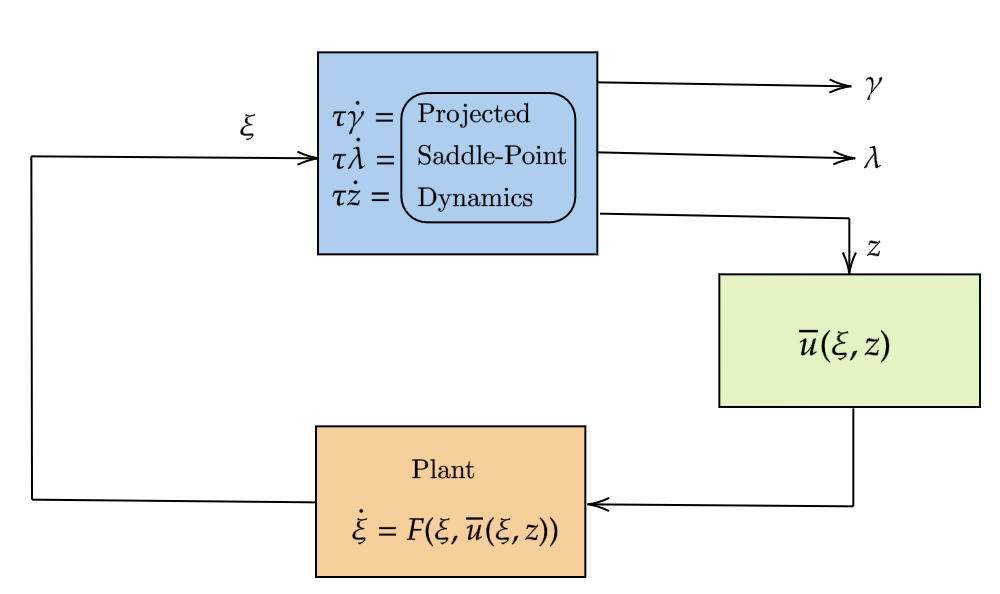}
    \caption{Block diagram of~\eqref{eq:alg-sp-plant}. The blue block updates the variables $\gamma$, $z$ and $\lambda$ using the projected saddle-point dynamics of the regularized version of~\eqref{eq:cbf-distr-opt-pb-y}. In parallel, the plant is updated using the controller $\bar{u}$.}
    \label{fig:scheme}
    \vspace{-3ex}
\end{figure} 

%To retain optimality guarantees, 
%
% \marginJC{What optimality guarantees? Of whom? I know what you mean, but the reader most probably won't}
%
% we do not want to keep $z$ fixed. 
%
% \marginJC{This is oddly expressed ("we do not want"). It's rather that we need to optimize over $z$ too, otherwise we don't have full optimality for (17). Maybe it's better to lay out for the reader what we have in mind right after (17): there are 2 sets of decision variables to worry about: the inputs and the introduced constraint-mismatch variables. For fixed latter ones, problem is (19), completely decoupled, can be solved "instantaneuly". To optimize the mismatch variables, we do (20)... Some box/table/algorithm environment might be visually helpful. The latter }
%
% In order to find an update law for $z$ that makes it converge to $z^{*,\epsilon}$, we consider the projected saddle-point dynamics of the regularized version of~\eqref{eq:cbf-distr-opt-pb-y} for a fixed value of $\xi$. 
We assume that the projected saddle-point dynamics can be run at a faster rate than the plant.
Introducing a timescale separation parameter $\tau>0$ to model this, the interconnection of the projected saddle-point dynamics with the plant leads to the following dynamical system:
\begin{subequations}
\begin{align}
  \tau\dot{\gamma}_i &= -\nabla_{\gamma_i} f_i(\xi_i, \gamma_i)- \sum_{ k\in P_i } \lambda_i^k \nabla_{\gamma_i} g_i^k (\xi_{\Nc_i}, \gamma_i),~\label{alg-v} \\
    \tau\dot{z}_i^k &=-\epsilon z_i^k -\sum_{j\in \Nc_i \cap \Gc_k} (\lambda_i^k-\lambda_j^k),~\label{alg-y} \\
    \tau\dot{\lambda}_i^k &= [g_i^k(\xi_{\Nc_i}, \gamma_i)+\sum_{j\in \Nc_i \cap \Gc_k}(z_i^k-z_j^k)]_{\lambda_i^k}^+~\label{alg-lambda}, \\
  \dot{\xi}_i &= F_i(\xi_i,\bar{u}_i(\xi_i,z_{\Nc_i})),~\label{alg-x}
\end{align}
\label{eq:alg-sp-plant}
\end{subequations}
for all $i\in[N]$ and $k\in P_i$. We henceforth denote $\gamma=[\gamma_1,\hdots,\gamma_N]$, $\lambda_i=\{ \lambda_i^k \}_{k\in P_i}$ for all $i\in[N]$ and $\lambda=
[\lambda_1,\hdots,\lambda_N]$.
The auxiliary variables $\gamma$ and $\lambda$ play the role of the primal variable $u$ and the Lagrange multipliers of the constraints in~\eqref{eq:cbf-distr-opt-pb-y}, respectively.
Our proposed algorithmic solution is the controller $\bar{u}_i$ for all $i\in[N]$ implemented for the current value of the state variables $\xi_i$ and $z_{\Nc_i}$ in~\eqref{eq:alg-sp-plant}. 

\section{Analysis of the Solution: Distributed Character, Safety, and
  Stability}\label{sec:analysis}

In this section we establish the properties of the controller proposed in Section~\ref{sec:proposed-solution}.
Throughout this section we use the same communication graph $\Gc$ introduced in Section~\ref{sec:proposed-solution}.
%
% \marginJC{If we don't have the formulation (18) in the section above (see my earlier margin), we could move this, along with the formulation, to an appendix. The statement of Prop 5.4 would need to be modified accordingly.}
%
First we introduce various assumptions regarding problem~\eqref{eq:cbf-distr-opt-pb-y-local-only-u} and discuss their sensibleness.

\begin{assumption}\longthmtitle{Feasibility of optimization problem}\label{as:feasibility-opt-pb}
    For all $i\in[N]$, the optimization problem~\eqref{eq:cbf-distr-opt-pb-y-local-only-u} is feasible for all $\xi_{\Nc_i}\in\real^{n |\Nc_i|}$ and $z_{\Nc_i}\in\real^{\sum_{j\in\Nc_i}|P_j| }$.
\end{assumption}
\smallskip

\begin{remark}\longthmtitle{Handling feasibility in practice}
    % Mention that if $y$s are at the optimizer and~\eqref{eq:distr-opt-pb} is feasible, then~\eqref{eq:local-optimization-pb} is feasible.
    Problem~\eqref{eq:cbf-distr-opt-pb-y-local-only-u} is feasible if $z_{\Nc_i}=z_{\Nc_i}^{*,\epsilon}(\xi)$:  $u_i^{*,\epsilon}(\xi)$ is a solution of~\eqref{eq:cbf-distr-opt-pb-y-local-only-u} because of Proposition~\ref{prop:equivalence} and Assumption~\ref{as:feasibility-no-constraint-mismatch-vars}.
   Hence, if $z_{\Nc_i}$ is close to $z_{\Nc_i}^{*,\epsilon}(\xi)$, then~\eqref{eq:cbf-distr-opt-pb-y-local-only-u} is often feasible.
    Moreover, in practice, one can tune $\alpha_{c}^{ij}$ and $\alpha_k$ constraints in~\eqref{eq:diff-drive-pb} only when they are close to being active to reduce the number of overall constraints and facilitate feasibility.
    \demo
\end{remark}

% \begin{remark}
%     All of this is a heuristic and there are no guarantees of convergence (feasibility).
% \end{remark}

\begin{assumption}\longthmtitle{Availability of optimizer in real time}\label{as:availability-optimizer}
    The function $\bar{u}_i$ is instantaneously available to agent $i$ for all $i\in[N]$.
\end{assumption}
\smallskip 

Assumption~\ref{as:availability-optimizer} is a reasonable abstraction of what happens in practical scenarios, such as~\eqref{eq:diff-drive-pb}, which is a quadratic program and can be solved efficiently~\cite{BS-GB-PG-AB-SB:20}.
In fact, if $M\leq 2$,
$\bar{u}$ can even be found in closed form~\cite[Theorem 1]{XT-DVD:24}.

\begin{assumption}\longthmtitle{Lipschitzness of optimizer}\label{as:lipschitzness-optimizer}
    The functions $u^{*,\epsilon}$, $z^{*,\epsilon}$ and $\bar{u}$ are locally Lipschitz.
\end{assumption}

\begin{remark}\longthmtitle{Conditions that ensure Lipschitzness}
    The works~\cite{BJM-MJP-ADA:15,PM-AA-JC:23-scl} study different conditions under which the solution of parametric optimization problems such as~\eqref{eq:cbf-distr-opt-pb-y-local-only-u} or~\eqref{eq:cbf-distr-opt-pb-y} is locally Lipschitz. 
    % Under twice continuous differentiability of the objective function and constraints, Slater's condition and the \textit{constant-rank constraint qualification} (resp. the linear independence of the gradients of the active constraints),~\cite{JL:95} (resp.~\cite{SMR:80}) shows that the solution of parametric optimization problems such as~\eqref{eq:cbf-distr-opt-pb-y-local-only-u} or~\eqref{eq:cbf-distr-opt-pb-y} is locally Lipschitz. 
    % Recently,~\cite{PM-AA-JC:23-scl} gives conditions that ensure the weaker notion of \textit{point-Lipschitzness}, which ensures existence of solutions of the closed-loop system.
    % Lots of works study conditions under which the solution of parametric optimization problems like~\ref{eq:cbf-distr-opt-pb-y-local-only-u} is locally Lipschitz (see e.g.,~\cite{SMR:80}). Moreover, in the context of CBF-based quadratic programs like~\eqref{eq:diff-drive-pb},~\cite{BJM-MJP-ADA:15} gives conditions under which the controllers obtained from solving such state-dependent optimization problems are continuous and smooth.
    %I wouldn't point out in this paper all the technical difficulties behind all of this. Instead, making this Lipschitzness assumption simplifies things a lot.
    \demo
\end{remark}

We next show that the controller $\bar{u}$ is \textbf{safe} and \textbf{asymptotically} converges to $u^{*,\epsilon}$ when implemented in conjunction with the projected saddle-point dynamics as in~\eqref{eq:alg-sp-plant}. This, together with its distributed character, means that it meets all the desired properties. For to the problem described in Section~\ref{sec:problem-statement}, this means that $\bar{u}$ achieves obstacle and inter-agent collision avoidance, and asymptotically converges to the closest controller to $u_{\text{nom}}=[u_{\text{nom},1},\hdots,u_{\text{nom},N}]$ that satisfies the safety constraints. In particular, if the agents are far from any of the obstacles in the environment and the CBF constraints in~\eqref{eq:diff-drive-pb} are inactive, $\bar{u}$ converges to $u_{\text{nom}}$ and steers the \textit{leader} towards the desired waypoint and the \textit{followers} towards their formation positions.

\begin{proposition}\longthmtitle{Convergence of algorithm}\label{prop:alg-convergence}
Suppose that Assumptions~\ref{as:feasibility-no-constraint-mismatch-vars}-\ref{as:lipschitzness-optimizer} hold. Then,
\begin{enumerate}
    \item\label{it:prop-first} the controller $\bar{u}$ is \textbf{safe}, i.e., if the initial conditions of~\eqref{eq:alg-sp-plant} are such that
    $g_i^k(\xi_{\Nc_i}(0),\bar{u}_i(\xi_i(0), z_{\Nc_i}(0)) )\leq 0$ for all $k\in[p]$, then
    the trajectories of~\eqref{eq:alg-sp-plant} satisfy
    \begin{align}\label{eq:ubar-safe}
        \sum_{i\in V(\Gc_k)} g_i^k(\xi_{\Nc_i}(t),\bar{u}_i(\xi_i(t), z_{\Nc_i}(t)) )\leq 0,
    \end{align}
    for all $k\in[p]$ and $t\geq0$;
    %
% \marginJC{The satisfaction of this inequality is really about the differential CBF condition, rather than the forward invariant of a set statement (the former implies the latter, but we're not been very explicit about it). Also, shouldn't it be a requirement that initial condition is safe?}
    %
    \item\label{it:prop-second} if the origin is asymptotically stable for the dynamical system $\dot{\xi}=F(\xi,u^{*,\epsilon}(\xi))$ with Lyapunov function $V:\real^n\to\real$,
    and $\Gamma$ is a Lyapunov sublevel set of $V$ contained in its region of attraction,
    then, for any initial condition of~\eqref{eq:alg-sp-plant} $c_0:=(\gamma_0,z_0,\lambda_0,\xi_0)\in\real^{nN}\times\real^{q}\times\real^{q}\times\real^{nN}$, with $\xi_0\in\Gamma$, 
    %to be precise we should probably take $\xi_0$ in a subset of $\Rc$ but we are just extending the results in [32] a bit informally
    $\delta>0$ and $\epsilon>0$, there exist $\tau_{c_0,\delta,\epsilon} > 0$ and $r_{\xi_0}$ such that, if $\max \{ \norm{\gamma_0-u^{*,\epsilon}(\xi_0)}, \norm{z_0-z^{*,\epsilon}(\xi_0)},  \norm{\lambda_0-\lambda^{*,\epsilon}(\xi_0)} \} < r_{\xi_0}$, and $\tau < \tau_{c_0,\delta,\epsilon}$, then the trajectories of~\eqref{eq:alg-sp-plant} are such that, for all $t>0$,
%
% \marginJC{Presumably, $\delta$ is much smaller than $r_{c_0,\delta}$, but there is no quantitative sense of it in the statement}
%
\begin{align*}
    &\norm{\gamma(t)\!-\!u^{*,\epsilon}(\xi(t))} \! \leq \! \delta, \ \norm{z(t)\!-\!z^{*,\epsilon}(\xi(t))} \! \leq \! \delta, 
    \\
    &\norm{\lambda(t) \! - \! \lambda^{*,\epsilon}(\xi(t))} \! \leq \! \delta, \ \norm{\bar{u}(\xi(t),z(t))\!-\! u^{*,\epsilon}(\xi(t))} \! \leq \! \delta,
\end{align*}
and there exists $T_{c_0,\delta,\epsilon}>0$ such that $\norm{\xi(t)} \leq \delta$ for all $t > T_{c_0,\delta,\epsilon}>0$.
\end{enumerate}
\end{proposition}
\begin{proof}
First we show~\ref{it:prop-first}. By definition of $\bar{u}$, we have
\begin{align}\label{eq:first-proof}
    g_i^k(\xi_{\Nc_i}(t), \bar{u}_i(\xi_i(t), z_{\Nc_i}(t) ) )+ \! \! \sum_{j\in \Nc_i \cap V(\Gc_k)} \! \!(z_i^k(t)-z_j^k(t)) \leq0,
\end{align}
for all $i\in[N]$, $k\in[p]$ and $t\geq0$. Adding~\eqref{eq:first-proof} for all $i\in V(\Gc_k)$, we obtain~\eqref{eq:ubar-safe} for all $k\in[p]$ and $t\geq0$.
Next we show~\ref{it:prop-second}.
Since the dynamics in~\eqref{eq:alg-sp-plant} are not differentiable due to the presence of the $[\cdot]_{+}$ operator, the standard version of Tikhonov's theorem for singular perturbations~\cite[Theorem 11.2]{HK:02} is not applicable. Instead we use~\cite{FW:05}, which gives a Tikhonov-type singular perturbation statement for differential inclusions.
For non-smooth ODEs like~\eqref{eq:alg-sp-plant} we need to check the following assumptions. 
First, the dynamics~\eqref{eq:alg-sp-plant} are well-defined because Assumption~\ref{as:feasibility-opt-pb} guarantees that $\bar{u}(\xi,z)$ is well-defined for all $\xi\in\real^{nN}$ and $z\in\real^{q}$.
Second, the dynamics~\eqref{eq:alg-sp-plant} are locally Lipschitz because of the Lipschitzness of the gradients of $f$ and $g$, and the $\textit{max}$ operator, as well as the Lipschitzness of $F_i$ for all $i\in\real^n$ and assumption~\ref{as:lipschitzness-optimizer}.
Third, the existence and uniqueness of the equilibrium of the fast dynamics
%
% \marginJC{Blue sentences are missing verbs, they don't read like complete sentences.}
%
follows from the fact that~\eqref{eq:cbf-distr-opt-pb-y-epsilon} has a strongly convex objective function and convex constraints, which implies that it has a unique KKT point.
Fourth, Lipschitzness and asymptotic stability of the reduced-order model
\begin{align}\label{eq:rom}
    \dot{\xi}=F(\xi,\bar{u}(\xi,z^{*,\epsilon}(\xi))).
\end{align}
Lipschitzness follows from Assumption~\ref{as:lipschitzness-optimizer}, and asymptotic stability follows from the fact that $\bar{u}(\xi,z^{*,\epsilon}(\xi))=u^{*,\epsilon}(\xi)$ for all $\xi\in\real^n$ (cf. Proposition~\ref{prop:equivalence}) and the hypothesis that the origin of $\dot{\xi}=F(\xi,u^{*,\epsilon}(\xi))$ is asymptotically stable.
Fifth, the asymptotic stability of the fast dynamics for every fixed value of the slow variable follows from~\cite[Theorem 5.1]{AC-EM-SHL-JC:18-tac}. Finally, the origin is the only equilibrium of~\eqref{eq:rom} by assumption.
The result follows from~\cite[Theorem 3.1 and Corollary 3.4]{FW:05}, by adapting the results therein to the case where the origin of~\eqref{eq:rom} has a bounded region of attraction. 
%We might want to add a sentence as to why the last inequality holds: it is because $\bar{u}(\xi,z^{*,\epsilon}(\xi))=u^{*,\epsilon}(\xi).$
%
% \marginJC{Is~\cite[Corollary 3.4]{FW:05} formulated like our statement here? I ask b/c I thought those singular perturbation results typically have that "limit as $\tau \rightarrow 0$. Just curious}
%
\end{proof}

%
% \marginJC{Shouldn't we comment somewhere (here?) that if (ii) holds globally, then a stronger result with bla, bla holds?}
%

If the constraints in~\eqref{eq:distr-opt-pb} correspond to the CBF conditions of some set, as it is the case in~\eqref{eq:diff-drive-pb}, Proposition~\ref{prop:alg-convergence}(i) implies that the set is forward invariant under the dynamics~\eqref{eq:alg-sp-plant}.
Moreover, if $\dot{\xi}=F(\xi,u^{*,\epsilon}(\xi))$ is globally asymptotically stable, then Proposition~\ref{prop:alg-convergence}(ii) holds for any $\xi_0\in\real^n$.
% We finalize this section with two remarks regarding the assumptions of Proposition~\ref{prop:alg-convergence}.

\begin{remark}\longthmtitle{Proximity of the constraint mismatch variables to their optimal values and choice of timescale}\label{rem:closeness-cm-optimizer}
    Proposition~\ref{prop:alg-convergence} requires that the initial conditions of $\gamma$,
    $z$ and $\lambda$ are close enough to $\gamma^{*,\epsilon}(\xi_0), z^{*,\epsilon}(\xi_0)$ and $\lambda^{*,\epsilon}(\xi_0)$. This is due to the technical nature of the proof of~\cite[Theorem 3.1]{FW:05}, which requires the fast variables to be in a small ball around the solution manifold to show the stability of the interconnected system.
    The satisfaction of these conditions can be achieved by exploiting the asymptotic stability properties of the projected saddle-point dynamics and running them for the regularized version of~\eqref{eq:cbf-distr-opt-pb-y}
    offline for a fixed value of the state $\xi$ equal to $\xi_0$ within an accuracy smaller than~$r_{\xi_0}$.
    Proposition~\ref{prop:alg-convergence} also requires $\tau$ to be sufficiently small. In practice we have observed that convergence of the state variables is achieved for a wide range of values of $\tau$. Moreover, since safety holds for any $\tau$, the value of $\tau$ can be decreased during the execution of the algorithm to ensure convergence to the desired waypoint.
    %
% \marginJC{What does "this" refer to? Also, shouldn't the initial condition $\xi_0$ be at least safe/feasible?}
    %
    \demo
\end{remark}

\begin{remark}\longthmtitle{Asymptotic stability assumption}\label{rem:asymptotic-stab-plant}
Proposition~\ref{prop:alg-convergence}~\ref{it:prop-second} requires that the controller $u^{*,\epsilon}$ is asymptotically stabilizing. In the context of the problem outlined in Section~\ref{sec:problem-statement}, this means that in a neighborhood of the origin,
enforcing the obstacle avoidance and inter-agent collision avoidance constraints does not disrupt the stabilizing character of the nominal controllers (i.e., their steering towards the desired waypoints or desired formation positions of interest), which can be achieved by taking the values of $\alpha_c^{ij}$ and $\alpha_k$ sufficiently large.
Recent work~\cite{WSC-DVD:22-tac,PM-JC:23-csl} gives conditions under which the solution of a CBF-based QP of the form~\eqref{eq:diff-drive-pb} with a nominal stabilizing controller retains its stability properties. Such conditions can be used to derive a subset contained in the region of attraction of the origin, in which the result in Proposition~\ref{prop:alg-convergence}~\ref{it:prop-second} can be applied.
    \demo
\end{remark}

\section{Experimental validation}\label{sec:experiments}

Here we show the performance of the proposed distributed control design~\eqref{eq:alg-sp-plant} in simulation and in physical robotic platforms for a team of differential-drive robots, cf. Section~\ref{sec:problem-statement}.

\subsection{Parameter Tuning}
Our experiments underscore the sensitivity of the control design to the choice of parameters. In particular, 
%even if all waypoints lie in safe areas, 
certain sequences of waypoints
might lead to some of the agents of the team reaching deadlocks near the obstacles. This is a well-known issue of CBF-based controllers, cf.~\cite{MFR-APA-PT:21,PM-JC:23-csl}. In practice, we have observed that this behavior can be avoided by choosing a sequence of waypoints such that straight-lines connecting consecutive waypoints lie in the safe region, and promoting larger values of the angular velocity input, which allow the vehicle more \textit{manoeuvrability}, by selecting the matrix $\Gamma_i$ as the constant matrix $[5, 0; 0, 1]$).
We keep the other design parameters constant across the different experiments, with values $l=0.2$, $\alpha_c^{ij}=2.0$ for all $i\in[N]$, $j\in N_i$, $\alpha_k=2.0$ for all $k\in[M]$, $d_{\min}=1.0$, $\eta_k=1.5$ for all $k\in[M]$, $\epsilon=0.001$, and $\tau=0.1$. Given the initial condition $x(0)$, we follow the procedure in Remark~\ref{rem:closeness-cm-optimizer} to initialize the variables $\gamma, z$ and $\lambda$ in~\eqref{eq:alg-sp-plant} before executing the controller. For each robot $i$, the set $N_i$ is taken as the two closest robots to agent $i$ in the initial positions. %Since the robots attempt to maintain a fixed formation, this set of two closest robots remains mostly constant throughout the execution of the controller.

\subsection{Experiments}
We have tested our algorithm in different simulation and hardware environments.
For the simulation environments we have employed a high-fidelity Unity simulator
on an Ubuntu Laptop with Intel Core i7-1355U (4.5 GHz). 
We numerically integrate~\eqref{eq:alg-sp-plant} and implement $\bar{u}$ using the convex optimization library \texttt{CVXOPT}~\cite{MSA-JD-LV:13}.
The first simulation environment consists of a series of red cilindrical, cubic, and spherical obstacles. A team of 5 Husky\footnote{Spec. sheets for the Husky and Jackal robots can be found at https://clearpathrobotics.com} robots
%
% \marginJC{Reference for Husky robots -- to Clearpath robotics website}
%
initially in a x-like formation traverses the environment while avoiding collision with obstacles and other robots of the team\footnote{Video of the simulation: https://tinyurl.com/distributedcbfs}. The simulated robots have the same LIDAR and sensor capabilities as the real ones, and these are used to run a SLAM system that allows each robot to localize itself in the environment and obtain its current state, which is needed to run~\eqref{eq:alg-sp-plant} and implement $\bar{u}$.
Figure~\ref{fig:cbf-env-snapshots} shows 3 snapshots of the experiment and the trajectories followed by the robots.
The robots successfully complete the task by avoiding collisions and reaching all the waypoints while maintaining the desired formation. Figure~\ref{fig:cbf_5huskies} shows the evolution of the distance to the desired formation position for each follower and the distance to the different obstacles for the leader agent.

\begin{figure*}[htb]
  \centering
  {\includegraphics[width=0.27\linewidth]{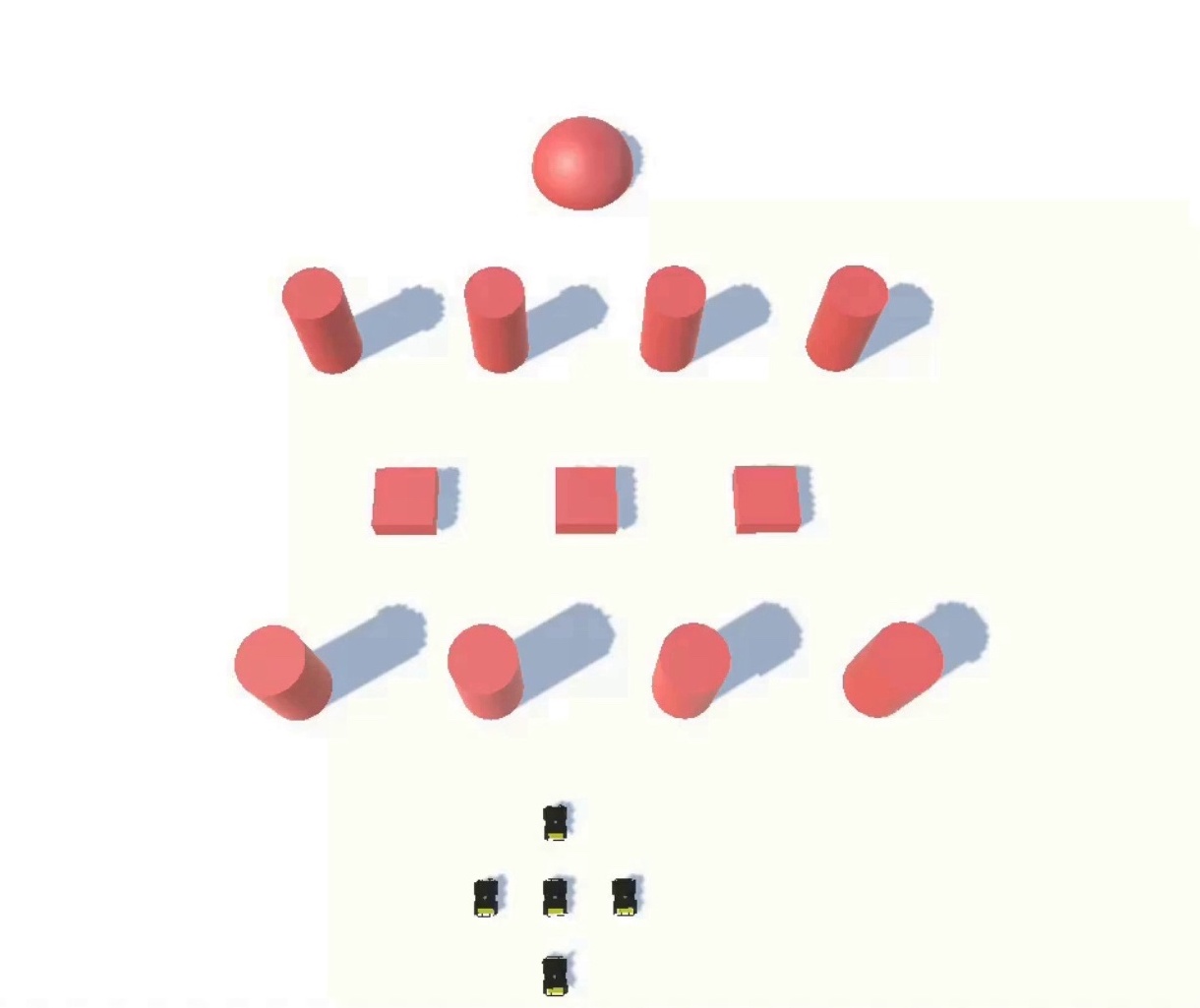}}
  \quad
  {\includegraphics[width=0.27\linewidth]{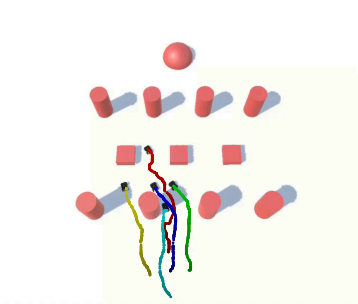}}
  \quad 
  {\includegraphics[width=0.27\linewidth]{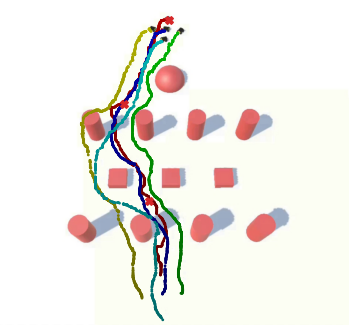}}
  \caption{Snapshots of the first simulation environment with color coded trajectories for the different robots. The intensity of the color decreases with time. In the last snapshot, the red x's indicate the three different waypoints for the leader of the team (in magenta). The environment has dimensions $20$m $\times$ $30$m.}
  \label{fig:cbf-env-snapshots}
\end{figure*}

\begin{figure}[htb]
  \centering
  \includegraphics[width=0.49\linewidth]{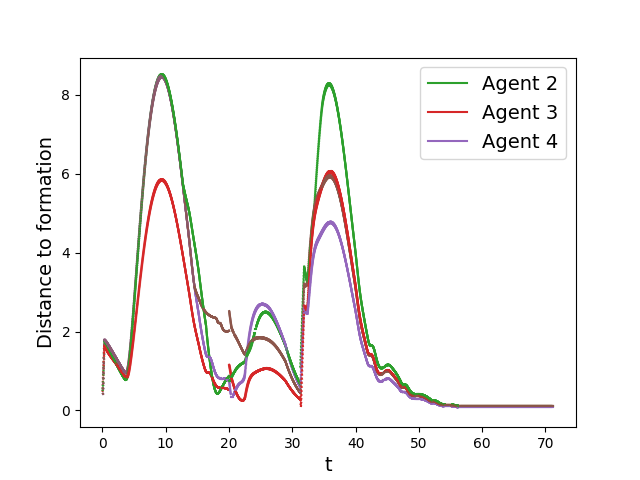}
  \includegraphics[width=0.49\linewidth]{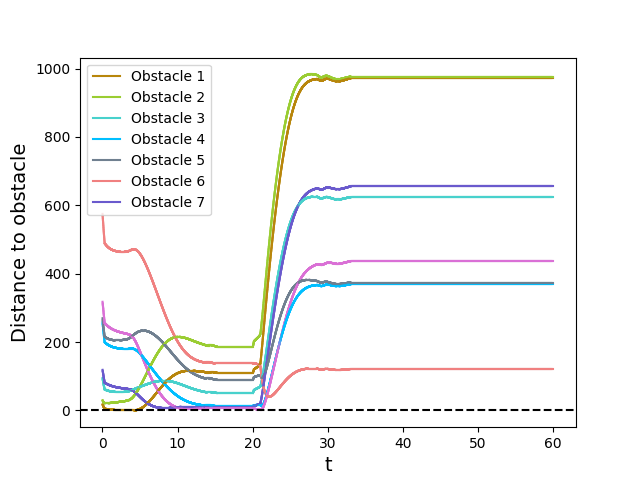}
  \caption{Evolution in the first simulation environment, cf. Figure~\ref{fig:cbf-env-snapshots}. (left): Distance to the desired formation position over time for the different followers. 
  %All followers asymptotically converge to their desired formation positions. 
  (right): Distance to the different obstacles over time for the \textit{leader}.}
  \label{fig:cbf_5huskies}
\end{figure}

In the second simulation, initially a team of three Husky robots are located in one of the rooms in the environment, cf. Figure~\ref{fig:office-env-snapshots}(left). The team traverses the environment while maintaining a triangular formation.
% by following the leader (the robot tracing a green trajectory). The leader follows three waypoints consecutively as seen in Figure~\ref{fig:office-env-snapshots}(right).
The walls are modelled as obstacles using a set of ellipsoidal barrier functions.
The team successfully completes the task by avoiding collisions while maintaining the desired formation.

\begin{figure*}[htb]
  \centering
  {\includegraphics[width=0.27\linewidth]{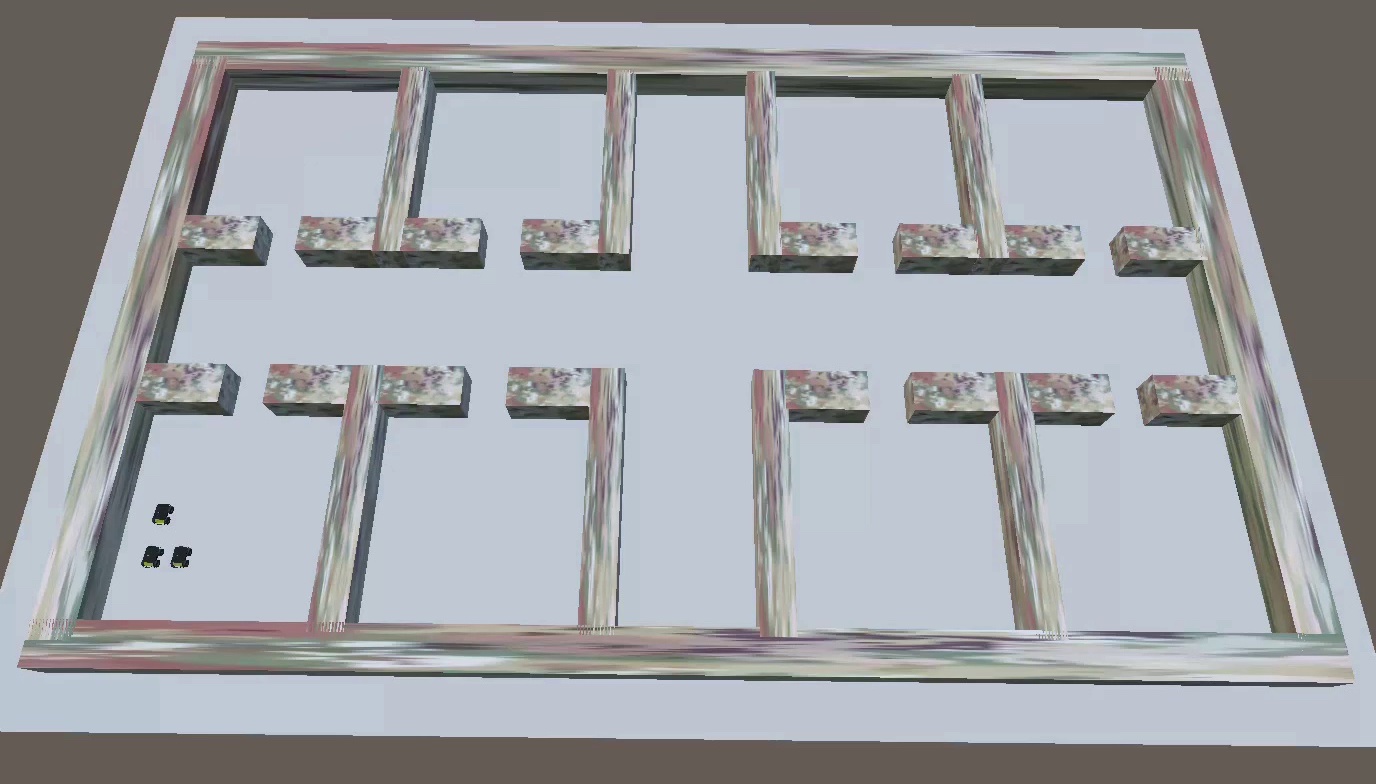}}
  \quad
  {\includegraphics[width=0.27\linewidth]{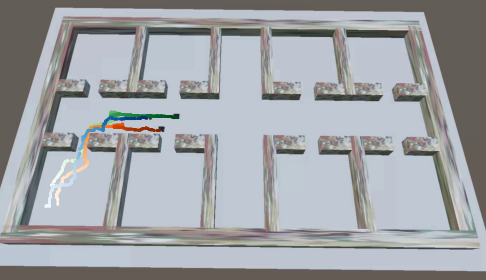}}
  \quad 
  {\includegraphics[width=0.27\linewidth]{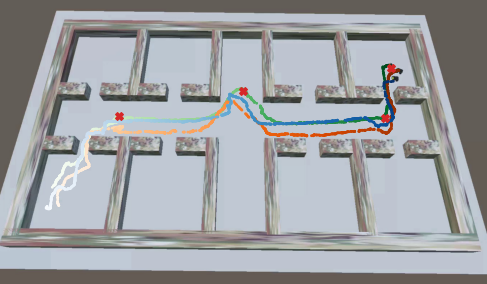}}
  \caption{Snapshots of the second simulation environment with color coded trajectories for the different robots. The intensity of the color increases with time. In the last snapshot, the red x's indicate the four different waypoints for the leader of the team. The environment has dimensions $20$m $\times$ $50$m.}
  \label{fig:office-env-snapshots}
\end{figure*}

We have also validated our design in hardware in a team of 3 Jackal\footnotemark[1] robots with GPS, IMU, and LIDAR sensors, which they use to run a SLAM system to localize itself in the environment.
%
% \marginJC{Provide a reference/spec sheet for the Jackals (e.g., link to clearpath robotics company)}
%
The team is initially positioned as shown in Figure~\ref{fig:hardware-snapshots}(left). The blue and grey cilinders are modelled as obstacles using CBFs. The team traverses the environment while maintaining a triangular formation.
% by following the leader (the robot tracing the green trajectory). 
%The leader follows two waypoints consecutively as seen in Figure~\ref{fig:hardware-snapshots}(right).
All computations are done onboard with the computers of each of the robots.
% The communication between the different robots of the team is done using a reliable UDP protocol.
%
% \marginJC{What does UDP stand for? Should we have a reference? Can one use an "unreliable" one instead?}
%

\begin{figure*}[htb]
  \centering
  {\includegraphics[width=0.27\linewidth]{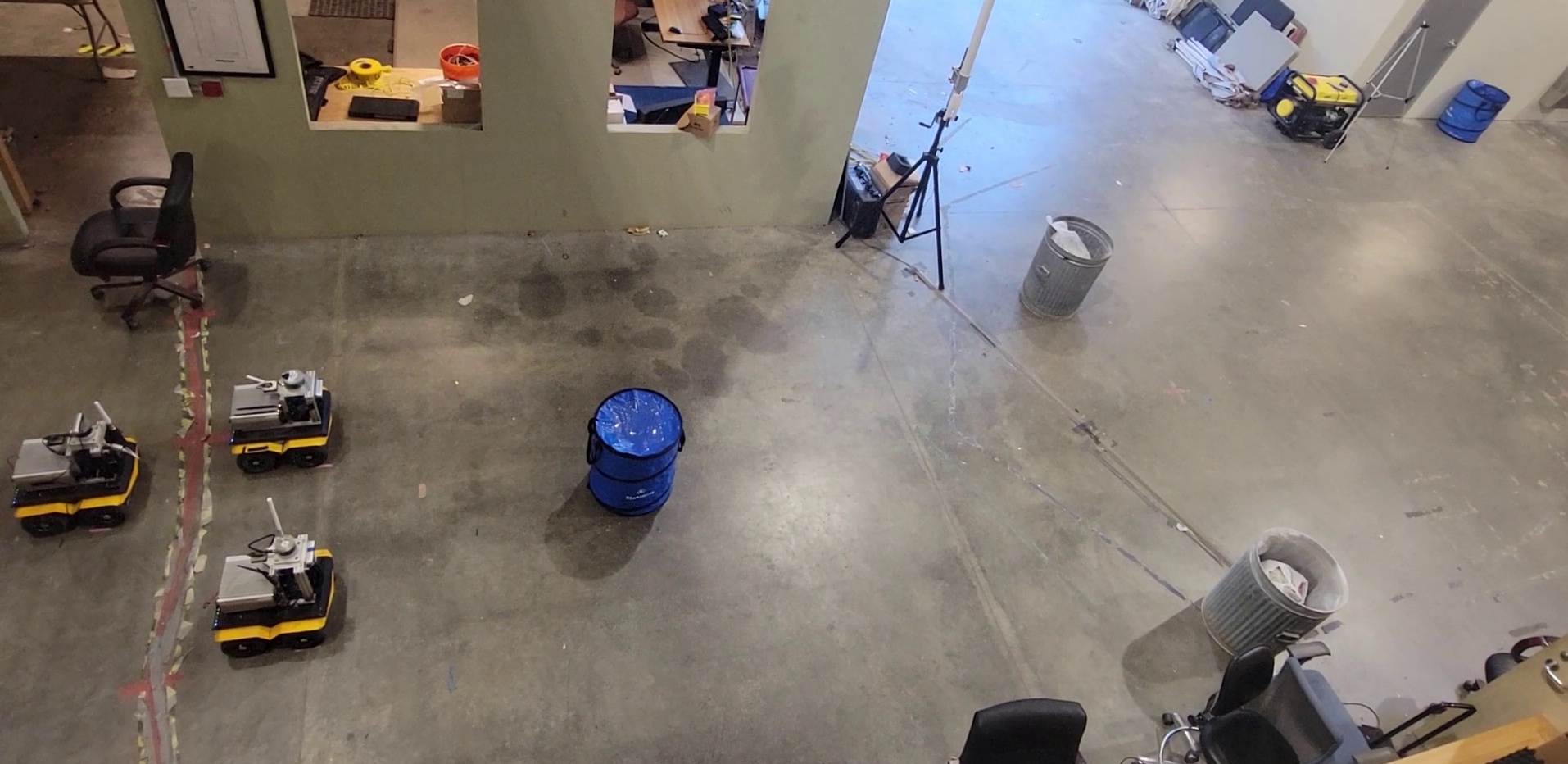}}
  \quad
  {\includegraphics[width=0.27\linewidth]{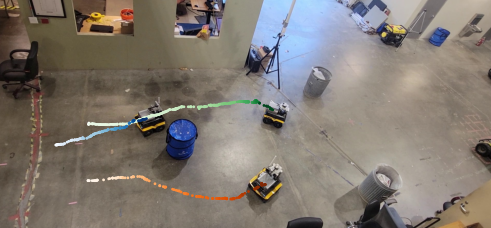}}
  \quad 
  {\includegraphics[width=0.27\linewidth]{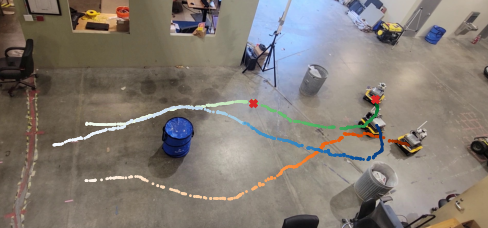}}
  \caption{Snapshots of the hardware experiment with color coded trajectories for the different robots. The intensity of the color increases with time. In the last snapshot, the red x's indicate the two different waypoints for the leader of the team. The environment has dimensions $4$m $\times$ $9$m.}
  \label{fig:hardware-snapshots}
\end{figure*}

% \begin{figure}[htb]
%   \centering
%   \includegraphics[width=0.5\textwidth]{figures/3jackals-507-slam.png}
%   \caption{Snapshot of the cameras and costmaps of the 3 Jackals in the hardware experiment}\label{3jackals-507-slam.png}
% \end{figure}

\section{Conclusions}\label{sec:conclusions}
We have proposed a distributed controller design based on the solution of state-dependent network optimization problems with coupling constraints under general assumptions. When interconnected with the network dynamics, our controller is guaranteed to converge to the solution of the network optimization problem and satisfy the constraints at all times. We have leveraged this framework to provide a distributed solution that achieves safe navigation for a team of differential-drive robots in environments with obstacles, avoiding collisions and maintaining a desired formation, using CBFs. We have illustrated its performance both in simulation and hardware. Future work will investigate conditions that ensure the feasibility of the optimization problem encoding the control design, determine explicit bounds on the timescale separation required for exact convergence, design sequences of waypoints that ensure the network optimization problem remains feasible,  investigate heuristics to tune the design parameters for improved performance, and explore the extension to settings with partial knowledge of the environment and the obstacles.

\section*{Acknowledgments}
This work was supported by the Tactical Behaviors for Autonomous Maneuver (TBAM) ARL-W911NF-22-2-0231. Pol Mestres did a
research internship at the U.S. Army Combat Capabilities Development
Command Army Research Laboratory in Adelphi, MD during the summer of 2023.

% \marginJC{Probably we want to change how [29] is displayed. No reference to SCL, only to archive?}

\bibliography{bib/alias,bib/JC,bib/Main-add,bib/Main,bib/New}

\begin{thebibliography}{10}
\providecommand{\url}[1]{#1}
\csname url@samestyle\endcsname
\providecommand{\newblock}{\relax}
\providecommand{\bibinfo}[2]{#2}
\providecommand{\BIBentrySTDinterwordspacing}{\spaceskip=0pt\relax}
\providecommand{\BIBentryALTinterwordstretchfactor}{4}
\providecommand{\BIBentryALTinterwordspacing}{\spaceskip=\fontdimen2\font plus
\BIBentryALTinterwordstretchfactor\fontdimen3\font minus
  \fontdimen4\font\relax}
\providecommand{\BIBforeignlanguage}[2]{{%
\expandafter\ifx\csname l@#1\endcsname\relax
\typeout{** WARNING: IEEEtran.bst: No hyphenation pattern has been}%
\typeout{** loaded for the language `#1'. Using the pattern for}%
\typeout{** the default language instead.}%
\else
\language=\csname l@#1\endcsname
\fi
#2}}
\providecommand{\BIBdecl}{\relax}
\BIBdecl

\bibitem{ADA-SC-ME-GN-KS-PT:19}
A.~D. Ames, S.~Coogan, M.~Egerstedt, G.~Notomista, K.~Sreenath, and P.~Tabuada,
  ``Control barrier functions: theory and applications,'' in \emph{{E}uropean
  {C}ontrol {C}onference}, Naples, Italy, 2019, pp. 3420--3431.

\bibitem{PW-FA:07}
P.~Wieland and F.~Allg{\"o}wer, ``Constructive safety using control barrier
  functions,'' \emph{IFAC Proceedings Volumes}, vol.~40, no.~12, pp. 462--467,
  2007.

\bibitem{ADA-XX-JWG-PT:17}
A.~D. Ames, X.~Xu, J.~W. Grizzle, and P.~Tabuada, ``Control barrier function
  based quadratic programs for safety critical systems,'' \emph{IEEE
  Transactions on Automatic Control}, vol.~62, no.~8, pp. 3861--3876, 2017.

\bibitem{RWB:91}
R.~W. Brockett, ``Dynamical systems that sort lists, diagonalize matrices, and
  solve linear programming problems,'' \emph{Linear Algebra and Its
  Applications}, vol. 146, pp. 79--91, 1991.

\bibitem{UH-JBM:94}
U.~Helmke and J.~B. Moore, \emph{Optimization and Dynamical Systems}.\hskip 1em
  plus 0.5em minus 0.4em\relax Springer, 1994.

\bibitem{AH-SB-FD:21}
A.~Hauswirth, S.~Bolognani, and F.~D{\"o}rfler, ``Projected dynamical systems
  on irregular, non-{E}uclidean domains for nonlinear optimization,''
  \emph{SIAM Journal on Control and Optimization}, vol.~59, no.~1, pp.
  635--668, 2021.

\bibitem{MC-EDA-AB:20}
M.~Colombino, E.~Dall'Anese, and A.~Bernstein, ``Online optimization as a
  feedback controller: Stability and tracking,'' \emph{IEEE Transactions on
  Control of Network Systems}, vol.~7, no.~1, pp. 422--432, 2020.

\bibitem{MC-JSP-AB:19}
M.~Colombino, J.~W. Simpson-Porco, and A.~Bernstein, ``Towards robustness
  guarantees for feedback-based optimization,'' \emph{{IEEE} Conf.\ on Decision
  and Control}, pp. 6207--6214, 2019.

\bibitem{AA-JC:24-tac}
A.~Allibhoy and J.~Cort\'es, ``Control barrier function-based design of
  gradient flows for constrained nonlinear programming,'' \emph{IEEE
  Transactions on Automatic Control}, vol.~69, no.~6, 2024, to appear.

\bibitem{AC-YC-BC-JC-EDA:24-tsg}
A.~Colot, Y.~Chen, B.~Corn\'elusse, J.~Cort\'es, and E.~Dall'Anese, ``Optimal
  power flow pursuit via feedback-based safe gradient flow,'' \emph{IEEE
  Transactions on Smart Grid}, 2024, submitted. Available at
  \url{https://arxiv.org/abs/2312.12267}.

\bibitem{UB-LW-ADA-ME:15}
U.~Borrmann, L.~Wang, A.~D. Ames, and M.~Egerstedt, ``Control barrier
  certificates for safe swarm behavior,'' \emph{IFAC-PapersOnLine}, vol.~48,
  no.~27, pp. 68--73, 2015.

\bibitem{LW-ADA-ME:17}
L.~Wang, A.~Ames, and M.~Egerstedt, ``Safety barrier certificates for
  collisions-free multirobot systems,'' \emph{IEEE Transactions on Robotics},
  vol.~33, no.~3, pp. 661--674, 2017.

\bibitem{XT-DVD:22}
X.~Tan and D.~V. Dimarogonas, ``Distributed implementation of control barrier
  functions for multi-agent systems,'' \emph{IEEE Control Systems Letters},
  vol.~6, pp. 1879--1884, 2022.

\bibitem{VNFA-XT-DVD:23}
V.~N. Fernandez-Ayala, X.~Tan, and D.~V. Dimarogonas, ``Distributed barrier
  function-enabled human-in-the-loop control for multi-robot systems,'' in
  \emph{{IEEE} Int. Conf.\ on Robotics and Automation}, London, UK, 2023, pp.
  7706--7712.

\bibitem{XT-CL-KHJ-DVD:24}
X.~Tan, C.~Liu, K.~H. Johansson, and D.~V. Dimarogonas, ``A continuous-time
  violation-free multi-agent optimization algorithm and its applications to
  safe distributed control,'' \emph{arXiv preprint arXiv:2404.07571}, 2024.

\bibitem{CL-XT-XW-DVD-KHJ:24}
C.~Liu, X.~Tan, X.~Wu, D.~V. Dimarogonas, and K.~H. Johansson, ``Achieving
  violation-free distributed optimization under coupling constraints,''
  \emph{arXiv preprint arXiv:2404.07609}, 2024.

\bibitem{PM-JC:23-cdc}
P.~Mestres and J.~Cort{\'e}s, ``Distributed and anytime algorithm for network
  optimization problems with separable structure,'' in \emph{{IEEE} Conf.\ on
  Decision and Control}, Singapore, 2023, pp. 5457--5462.

\bibitem{NDC-PS-PRG:24}
N.~de~Carli, P.~Salaris, and P.~R. Giordano, ``Distributed control barrier
  functions for global connectivity maintenance,'' in \emph{{IEEE} Int. Conf.\
  on Robotics and Automation}, Yokohama, Japan, 2024.

\bibitem{AC-EM-SHL-JC:18-tac}
A.~Cherukuri, E.~Mallada, S.~H. Low, and J.~Cort\'{e}s, ``The role of convexity
  in saddle-point dynamics: Lyapunov function and robustness,'' \emph{IEEE
  Transactions on Automatic Control}, vol.~63, no.~8, pp. 2449--2464, 2018.

\bibitem{AC-JC:16-allerton}
A.~Cherukuri and J.~Cort\'{e}s, ``Distributed algorithms for convex network
  optimization under non-sparse equality constraints,'' in \emph{Allerton
  Conf.\ on Communications, Control and Computing}, Monticello, IL, Sep. 2016,
  pp. 452--459.

\bibitem{PG-IB-ME:19}
P.~Glotfelter, I.~Buckley, and M.~Egerstedt, ``Hybrid nonsmooth barrier
  functions with applications to provably safe and composable collision
  avoidance for robotic systems,'' \emph{IEEE Robotics and Automation Letters},
  vol.~4, no.~2, pp. 1303--1310, 2019.

\bibitem{MA-GC-AB-AB:95}
M.~Aicardi, G.~Casalino, A.~Bicchi, and A.~Balestrino, ``Closed {L}oop
  {S}teering of {U}nicycle-like {V}ehicles via {L}yapunov {T}echniques,''
  \emph{{IEEE} Robotics and Automation Magazine}, vol.~2, no.~1, pp. 27--35,
  1995.

\bibitem{PM-JC:23-csl}
P.~Mestres and J.~Cort\'es, ``Optimization-based safe stabilizing feedback with
  guaranteed region of attraction,'' \emph{IEEE Control Systems Letters},
  vol.~7, pp. 367--372, 2023.

\bibitem{XX:18}
X.~Xu, ``Constrained control of input-output linearizable systems using control
  sharing barrier functions,'' \emph{Automatica}, vol.~87, pp. 195--201, 2018.

\bibitem{XT-DVD:22-cdc}
X.~X.~Tan and D.~Dimarogonas, ``Compatibility checking of multiple control
  barrier functions for input constrained systems,'' in \emph{{IEEE} Conf.\ on
  Decision and Control}, 2022, pp. 939--944.

\bibitem{BS-GB-PG-AB-SB:20}
B.~Stellato, G.~Banjac, P.~Goulart, A.~Bemporad, and S.~Boyd, ``{OSQP:} an
  operator splitting solver for quadratic programs,'' \emph{Mathematical
  {P}rogramming {C}omputation}, vol.~12, pp. 637--672, 2020.

\bibitem{XT-DVD:24}
X.~Tan and D.~V. Dimarogonas, ``On the undesired equilibria induced by control
  barrier function based quadratic programs,'' \emph{Automatica}, vol. 159, p.
  111359, 2013.

\bibitem{BJM-MJP-ADA:15}
B.~J. Morris, M.~J. Powell, and A.~D. Ames, ``Continuity and smoothness
  properties of nonlinear optimization-based feedback controllers,'' in
  \emph{{IEEE} Conf.\ on Decision and Control}, Osaka, Japan, Dec 2015, pp.
  151--158.

\bibitem{PM-AA-JC:23-scl}
P.~Mestres, A.~Allibhoy, and J.~Cort\'es, ``Robinson’s counterexample and
  regularity properties of optimization-based controllers,'' \emph{Preprint.
  Available at \url{https://arxiv.org/abs/2311.13167}}, 2023.

\bibitem{HK:02}
H.~Khalil, \emph{Nonlinear Systems, 3rd ed.}\hskip 1em plus 0.5em minus
  0.4em\relax Englewood Cliffs, NJ: Prentice Hall, 2002.

\bibitem{FW:05}
F.~Watbled, ``On singular perturbations for differential inclusions on the
  infinite interval,'' \emph{Journal of Mathematical Analysis and
  Applications}, vol. 310, no.~2, pp. 362--378, 2005.

\bibitem{WSC-DVD:22-tac}
W.~S. Cortez and D.~V. Dimarogonas, ``On compatibility and region of attraction
  for safe, stabilizing control laws,'' \emph{IEEE Transactions on Automatic
  Control}, vol.~67, no.~9, pp. 7706--7712, 2022.

\bibitem{MFR-APA-PT:21}
M.~F. Reis, A.~P. Aguilar, and P.~Tabuada, ``Control barrier function-based
  quadratic programs introduce undesirable asymptotically stable equilibria,''
  \emph{IEEE Control Systems Letters}, vol.~5, no.~2, pp. 731--736, 2021.

\bibitem{MSA-JD-LV:13}
M.~S. Andersen, J.~Dahl, and L.~Vandenberghe, ``{CVXOPT}: A python package for
  convex optimization, version 1.1.6.'' \emph{Available at cvxopt.org, 2013},
  2013.

\end{thebibliography}
\bibliographystyle{IEEEtran}

\appendix

The regularized version of~\eqref{eq:cbf-distr-opt-pb-y} takes the following form:
\begin{align}\label{eq:cbf-distr-opt-pb-y-epsilon}
  & \min \limits_{u\in\real^{mN}, z\in\real^q} \sum_{i=1}^N f_i(\xi_i, u_i) + \epsilon \sum_{k=1}^p \sum_{j\in V(\Gc_k)} (z_j^k)^2,
  \\
  \notag
    & \text{s.t.} \quad g_i^k(\xi_{\Nc_i}, u_i)+ \! \! \! \! \! \! \! \sum_{j\in \Nc_i\cap V(\Gc_k)} \! \! \! \! (z_i^k \! - \! z_j^k) \leq0, \ i \! \in \! V(\Gc_k), \ k \! \in \! [p].
\end{align}

The next sensitivity result shows that the solution to~\eqref{eq:cbf-distr-opt-pb-y-epsilon} and~\eqref{eq:cbf-distr-opt-pb-y} can be made close by choosing $\epsilon$ sufficiently small.

\begin{lemma}\longthmtitle{Sensitivity of regularized problem}\label{lem:sensitivity}
    Let $u^*$ be continuous, and assume $u^{*,\epsilon}$ is continuous for all $\epsilon>0$.
    Let $\Kc\subset\real^{nN}$ be compact
    and $\delta>0$. Then, there exists 
    $\bar{\epsilon}_{\Kc,\delta}>0$ such that if $\epsilon<\bar{\epsilon}_{\Kc,\delta}$, then $\norm{u^{*,\epsilon}(\xi)-u^*(\xi)}\leq \delta$ for all $\xi\in\Kc$.
\end{lemma}
\begin{proof}
    By~\cite[Lemma 4.2]{PM-JC:23-cdc}, for each $\xi\in\Kc$, there exists $\bar{\epsilon}_{\xi,\delta}>0$ such that if $\epsilon < \bar{\epsilon}_{\xi,\delta}$, then $\norm{u^{*,\epsilon}(\xi)-u^*(\xi)}\leq \frac{\delta}{2}$. Now, since $u^{*,\epsilon}$ and $u^*$ are continuous, there exists a neighborhood $\Nc_{\xi}$ of $\xi$ such that $\norm{u^{*,\epsilon}(\hat{\xi})-u^*(\hat{\xi})}\leq \delta$ for all $\hat{\xi} \in \Nc_{\xi}$.
    Since $\cup_{\xi\in \Kc} \Nc_{\xi}$ is an open covering of the compact set $\Kc$, there exists a finite subcover, i.e., there exists $N_{\Kc}\in\mathbb{Z}_{>0}$ and $\{ \xi_i \}_{i=1}^{N_{\Kc}}$ such that $\Kc \subset \cup_{i=1}^{N_{\Kc}} \Nc_{\xi_i}$. The result follows by letting $\bar{\epsilon}_{\Kc,\delta}:=\min\limits_{i\in[N_{\Kc}]} \{ \bar{\epsilon}_{\xi_i,\delta} \}$.
\end{proof}

% Lemma~\ref{lem:sensitivity} ensures that in all of the results in Section~\ref{sec:proposed-solution}, one can assume a fixed value of $\epsilon$ has been chosen sufficiently small to guarantee a desired maximum distance between the optimizers of the regularized and unregularized problems in a compact set of interest.

\end{document}